\documentclass[11pt]{amsart}
\usepackage[T1]{fontenc}
\usepackage{amsmath}
\usepackage{amsfonts}
\usepackage{amssymb}
\newtheorem{thm}{Theorem}[section]

\newtheorem{cor}[thm]{Corollary}

\newtheorem{prop}[thm]{Proposition}
\newtheorem{prob}[thm]{Problem}
\theoremstyle{definition}
\newtheorem{defn}[thm]{Definition}
\newtheorem{remark}[thm]{Remark}

\usepackage[colorlinks=true,linkcolor=black,citecolor=black,urlcolor=black]{hyperref}

\DeclareMathOperator{\re}{Re}
\DeclareMathOperator{\rk}{rank}
\DeclareMathOperator{\Perm}{Perm}

\DeclareMathOperator{\Tr}{Tr}
\DeclareMathOperator{\Cut}{Cut}
\DeclareMathOperator{\DCut}{DCut}
\DeclareMathOperator{\conv}{conv}

\newcommand{\bb}[1]{\mathbb{#1}}
\newcommand{\cl}[1]{\mathcal{#1}}

\begin{document}

\title[]{Unique Games and Games Based on Groups}

\author[V.~I.~Paulsen]{Vern I.~Paulsen}
\address{Institute for Quantum Computing and Department of Pure Mathematics, University of Waterloo,
Waterloo, ON, Canada  N2L 3G1}
\email{vpaulsen@uwaterloo.ca}
\author[R.~H.~Levene]{Rupert H. Levene}
\address{School of Mathematics \&\ Statistics and Centre for Quantum Engineering, Science \&\ Technology, University College Dublin, Dublin 4, Ireland}
\email{rupert.levene@ucd.ie}
\thanks{The author(s) would like to thank the Isaac Newton Institute for Mathematical Sciences, Cambridge, for support and hospitality during the programme {\it Quantum information, quantum groups and operator algebras} where work on this paper was undertaken. This work was supported by EPSRC grant no EP/K032208/1.}

\begin{abstract} We study unique games and estimate some of their values. We prove that if a unique game has a quantum-assisted value close to 1, then it must have a perfect deterministic strategy. We introduce a family of unique games based on groups that generalize XOR games, and show that when the group is the cyclic group of order 3, then these games correspond to a 3-labelling problem for directed graphs.
\end{abstract}
\maketitle

\section{Introduction}

 This paper is concerned with a certain family of two person games, called {\it unique games}, and a special subclass of unique games, which we call {\it games based on groups}. When the group is the cyclic group of order two, then these group based games correspond to the well studied family of XOR games. We show that a certain labelling problem for directed graphs corresponds to a family of group based games for the cyclic group of order 3.
 
 We are especially interested in the case that there is a probability density on the inputs of these games, and the players are allowed to use probabilistic strategies. In these cases the question always arises of whether or not a quantum-assisted strategy can yield a higher expected value of winning the game than any classical strategy. In the case of the labelling problem for directed graphs, we show that the classical value of winning can be interpreted as a directed analogue of the cut value for undirected graphs and that often the quantum-assisted value of this directed cut number is higher.
 
 Intuitively, in these games the two players, who are often called Alice and Bob, are playing cooperatively and attempting to give correct pairs of answers to pairs of questions posed by a third party often called the {\it Referee} or {\it Verifier}.
 Whether the pair of answers returned by the players is satisfactory or not depends not just on the individual answers but on the 4-tuple consisting of the question-answer pair.

 Such a game is described by two {\it input} sets $X, Y$, often thought of as sets of questions for Alice and Bob, respectively, two {\it output} sets $A, B$, often thought of as Alice and Bob's respective answer sets,
 and a function, often called the {\it rules} or {\it verification function},
 \[ V: X \times Y \times A \times B \to \{ 0, 1\},\]
 where
 \[W:= \{ (x,y,a,b) : V(x,y,a,b) = 1 \},\]
 is the set of {\it correct} or {\it winning} 4-tuples and
 \[ L:= \{ (x,y,a,b): V(x,y,a,b) =0 \},\]
 is the set of {\it incorrect} or {\it losing} 4-tuples.

 Thus, a game $\cl G$ is specified by $(X, Y, A, B, V)$.
 Before the game begins Alice and Bob know the four sets and the function $V$.
 Even though Alice and Bob are not allowed to communicate during the game they are allowed to communicate before the
 game and decide on some type of strategy.

 A two person game is called {\it synchronous} provided that $X=Y, \, A=B$ and the rule function satisfies the condition that $V(x,x,a,b) =0$ whenever $a \ne b$. Thus, if both players are ever asked the same question, then they must give the same reply in order to win.  A synchronous game is called {\it symmetric} provided that $V(x,y,a,b) = V(y,x,b,a)$.

A  game is called a {\it unique game} if for each triple $x \in X, y \in Y, a \in A$ there is a unique $b \in B$ such that $V(x,y,a,b) =1$ and, in addition, for each triple $x \in X, y \in Y, b \in B$ there is a unique $a \in A$ such that $V(x,y,a,b) =1$. Alternatively, since these hypotheses imply that $|A|=|B|$, we may assume that $A=B$ and that there is a map $f: X \times Y \to \Perm(A),$ where $\Perm(A)$ denotes the group of permutations of $A$, such that
\[ V(x,y,a,b) = 1 \iff a = f(x,y)b.\]
It is this characterization that we shall most often use for describing unique games and, consequently we will often describe a unique game as $\cl G=(X, Y, A, f)$, with the understanding that $A=B$,  $f: X \times Y \to \Perm(A)$ and that the verifier is the function given as above.

It is not hard to see that a unique game described in this manner is synchronous if and only if $X=Y$ and $f(x,x)$ is the identity permutation for all $x$, and is symmetric if and only if, in addition, $f(y,x) = f(x,y)^{-1}$.

A special type of unique game that we shall be interested in is when the set $A= \Omega$, with $\Omega$ a finite group and the only permutations that we allow are the ones that come from the left action of the group upon itself. That is, the verifier has the form $f:X \times Y \to \Omega$ and
\[ V(x,y,a,b) =1 \iff ab^{-1} = f(x,y).\] We will call such a game a {\bf group based game}.

One famous example of a group based game is the CHSH game which is the case that $X=Y= \Omega = \bb Z_2$, the binary field,
and
\[ V(x,y,a,b) =1 \iff a -b = xy:= f(x,y).\]

We will see later that when $\Omega$ is the cyclic group of order 3, then synchronous, symmetric group based games correspond to a certain type of labelling game for directed graphs.

\section{Correlations and Values of Games}

 When we want to talk about the probability of winning  a game $\cl G= (X, Y, A, B, V)$ we also need to specify a probability density on input pairs, i.e., a function $\pi: I_A \times I_B \to [0,1]$ such that
 \[ \sum_{x \in I_A, y \in I_B} \pi(x,y) =1.\]
 For games with densities, Alice and Bob also know the density before the start of the game.
 We write $(\cl G, \pi)$ to denote a game with density.

 A {\it deterministic strategy} for a game is a pair of functions $h_A: X \to A, \,\, h_B:Y \to B$ such that whenever Alice and Bob receive input pair $(x,y)$ they respond with output pair $(f(x), g(y))$.

 Given a game with a density $(\cl G, \pi)$ and a deterministic strategy as above, the probability of winning the game is
 \[ \sum_{x,y} \pi(x,y) V(x, y, h_A(x), h_B(y))= \sum_{x,y\colon(x,y,h_A(x),h_B(y)) \in W} \pi(x,y).\]

 The {\it classical value} or {\it deterministic value} of $(\cl G, \pi)$ is defined to be the supremum of this winning probability over all possible deterministic strategies, i.e., pairs of functions $h_A, h_B$. Computing this value often requires searching over all pairs of functions $h_A, h_B$ and so can be computationally difficult. %

 The classical value of a game is generally denoted $\omega(\cl G, \pi)$, but we shall write $\omega_{det}(\cl G, \pi)$, since we will be examining several different versions of the value of a game and want a notation that distinguishes them.

Suppose now that instead of a deterministic strategy, Alice and Bob return random answers to their question pairs. If we let
$p(a,b|x,y)$ denote the probability that they return the answer pair $(a,b)$ when given the question pair $(x,y)$, then over many rounds of the game the expected probability that they will win $(\cl G, \pi)$ is given by
\[ \bb E_{\cl G, \pi}(p) := \sum_{(x,y,a,b) \in W} \pi(x,y) p(a,b|x,y).\]

Given a game $\cl G$ a density $p(a,b|x,y)$ is called {\it perfect} provided the probability that it returns a wrong answer is 0, that is, provided that
\[ V(x,y,a,b) =0 \implies p(a,b|x,y) =0.\]
Note that if a density is perfect, then for any density $\pi$ on inputs, one has $\bb E_{\cl G, \pi}(p) =1.$

The various values of a game that are studied are defined by taking the supremum of $\bb E_{\cl G, \pi}(p)$ over certain sets of conditional probabilities.  These sets of conditional probabilities are called {\it local, quantum, quantum spatial, quantum approximate, quantum commuting,} and {\it nonsignalling} and are denoted,
\[ C_t(n_A, n_B, k_A, k_B),\]
with $t = loc, q, qs, qa, qc, ns$, respectively, where $n_A=|X|$ and $n_B=|Y|$ indicate the number of Alice's and Bob's inputs and $k_A=|A|$ and $k_B=|B|$ indicate the number of their outputs.
Usually, for the sake of convenience we will omit $n_A, n_B, k_A, k_B$ in our notations.

These sets are all convex and they satisfy,
\[ C_{loc} \subsetneq C_q \subsetneq C_{qs} \subsetneq C_q^-= C_{qs}^- = C_{qa} \subsetneq C_{qc} \subsetneq C_{ns},\]
where the notation $C_q^-$ and $C_{qs}^-$ indicates the closure of the set viewed as a subset of $\bb R^{n_An_Bk_Ak_B}$.

We refer the reader to \cite{HMNPR, HP} for the precise definitions of these sets and for references to the papers that show that the containments are all strict.

The corresponding values of a game with density $(\cl G, \pi)$ are obtained by taking the supremum over each of these sets and these suprema are denoted by, respectively, $\omega_t(\cl G, \pi)$ with $t= loc, q, qs, qa, qc, ns$.
Because of the closure properties,  $\omega_q(\cl G, \pi) = \omega_{qs}(\cl G, \pi) = \omega_{qa}(\cl G, \pi)$. Moreover because of the containments,
\[ \omega_{loc}(\cl G, \pi) \le \omega_q(\cl G, \pi) \le \omega_{qc}(\cl G, \pi) \le \omega_{ns}(\cl G, \pi).\]
It is also known that $\omega_{loc}(\cl G, \pi) = \omega_{det}(\cl G, \pi)$, since the extreme points of the set $C_{loc}$ correspond to densities of the form
\[ p(a,b|x,y) = \begin{cases} 1 & a=h_A(x), b= h_B(y).\\ 0 & \text{else}, \end{cases} \]
for some pair of functions.

In the case that $X=Y$ and $A=B$, we call a density {\it synchronous} provided that
\[ p(a,b| x,x) =0, \,\, \forall a \ne b, \, \forall x,\]
and we call a density {\it symmetric} if
\[ p(a,b|x,y)=p(b,a|y,x),\,\,\forall x,y,a,b.\]
For $t= loc, q, qs, qa, qc, ns$ we let $C_t^s(n,n,k,k) \subseteq C_t(n,n,k,k)$ denote the subset of synchronous densities, and often abbreviate this as $C_t^s$.
Moreover, if $\cl G$ is a game with $X=Y, \, A=B$ we let $\omega_t^s(\cl G, \pi)$ denote the supremum over all synchronous densities of the various types.

It was shown in \cite{KPS} that $(C_q^s)^- = C_{qa}^s$, so for any such game, we also have that
\[ \omega^s_q(\cl G, \pi) = \omega_{qs}^s(\cl G, \pi) = \omega_{qa}^s(\cl G, \pi).\]
In addition, it is readily seen that
\[ \omega^s_{loc}(\cl G, \pi) = \sup_h \{ \sum_{x,y} \pi(x,y) V(x, y,h(x), h(y)) \mid h:X \to A \}.\]

 Synchronous densities are known to arise from various types of traces, see \cite{KPS, HMNPR, HP} for precise statements.

\begin{remark} Every unique game has a perfect nonsignalling strategy. If $\cl G=(X,Y,A,f)$ is a unique game and we set
 \[ p(a,b|x,y) = \begin{cases} \frac{1}{|A|}, & a=f(x,y)b \\ 0 & else. \end{cases}\]
 Then this is a non-signalling density with $p_A(a|x) = p_B(b|y) = \frac{1}{|A|}$ and is a perfect strategy for $\cl G$.

 Moreover, if $\cl G$ is a synchronous unique game, then this density is synchronous, and if $\cl G$ is also symmetric, then this density is symmetric and synchronous.
 \end{remark}

There are two other sets of conditional probability densities that we wish to consider.

The notation for the first set that we shall use was introduced in \cite{PT}. This set of densities also appears elsewhere with different names and is related to the NPA-hierarchy. Depending on how the NPA-hierarchy is numbered, these are sometimes said to arise from the 1/2-level.

A density $p(a,b|x,y)$ is called a {\it vect-density}, if there is a Hilbert space $\cl H$, and sets of vectors $\{ v_{x,a}: x \in X, a \in A \}, \, \{ w_{y,b}: y \in Y, b \in B \}$ in $\cl H$ and a unit vector $\eta \in \cl H$ such that:
\begin{itemize}
\item $v_{x,a} \perp v_{x,a'}, \forall x, \forall a \ne a'$,
\item $w_{y,b} \perp w_{y, b'}, \forall y, \forall b \ne b'$,
\item $\sum_a v_{x,a} = \sum_b w_{y,b} = \eta, \forall x,y$,
\item $p(a,b|x,y) = \langle v_{x,a} | w_{y,b} \rangle \ge 0, \forall x,y,a,b$.\end{itemize}
The set of all such densities is denoted $C_{vect}$ and it is known to be a convex set satisfying
\[ C_{qc} \subset C_{vect} \subset C_{ns}, \] see \cite{PT}.
 Note that the condition that the inner products are all non-negative implies that without loss of generality, $\cl H$ can be assumed to be a real Hilbert space with a dimension that is bounded by a function of $n_A, k_A, n_B, k_B$.

In the case that $X=Y, \, A=B$, we show below that $p \in C^s_{vect}$ if and only if $v_{x,a} = w_{x,a}, \, \forall x,a$.

 We let $\omega_{vect}(\cl G, \pi)$ denote the supremum of $\bb E_{\cl G, \pi}(p)$ over all $p \in C_{vect}$, so that we have
 \[ \omega_{qc}(\cl G, \pi) \le \omega_{vect}(\cl G, \pi) \le \omega_{ns}(\cl G, \pi).\]
 The quantity $\omega_{vect}(\cl G, \pi)$ is often useful for giving a computable upper bound on $\omega_{qc}(\cl G, \pi)$.
 The quantity $\omega_{vect}^s(\cl G, \pi)$ is defined similarly.

 The next set of densities that we wish to define is useful for giving lower bounds on $\omega_q(\cl G, \pi)$.
 Recall that a density $p(a,b|x,y)$ is in $C_q$ if and only if there are finite dimensional Hilbert spaces $\cl H_A, \cl H_B$, projections $\{ E_{x,a} : x \in X, a \in A \} \subseteq B(\cl H_A)$ and $\{ F_{y,b}: y \in Y, b \in B \} \subseteq B(\cl H_B)$ with $ \sum_a E_{x,a} = I_{\cl H_A}, \forall x$ and $\sum_b F_{y,b} = I_{\cl H_B}, \forall y$ and a unit vector $\eta \in \cl H_A \otimes \cl H_B$ such that
 \[ p(a,b| x,y) = \langle E_{x,a} \otimes F_{y,b} \eta | \eta \rangle.\]
 We say that a density $p \in C_q$ is in $C_{q_1}$ if it is of this form with the additional requirement that
 \[ \rk(E_{x,a}) \le 1 \text{ and } \rk(F_{y,b}) \le 1,  \forall x,y,a,b.\]
 We let $\omega_{q_1}( \cl G, \pi)$(respectively, $\omega_{q_1}^s(\cl G, \pi)$), denote the supremum of $\bb E_{\cl G, \pi}(p)$ over all $p \in C_{q_1}$(respectively, $p \in C_{q_1}^s$).

 Note that the rank condition puts an upper bound on the dimensions of the Hilbert spaces for $q_1$-densities, since
 \[ \dim(\cl H_A) \le |A|, \, \dim(\cl H_B) \le |B|.\]
 Moreover, each of Alice's projections must  be of the form $E_{x,a} = |v_{x,a}\rangle \langle v_{x,a} | = v_{x,a} v_{x,a}^*$ where each of the vectors is either a unit vector or 0, $v_{x,a} \perp v_{x,a'}, \forall a \ne a'$ and the number of non-zero vectors is equal to $\dim(\cl H_A)$, with similar statements for $F_{y,b}$ and finally,
 \[ p(a,b|x,y) = | \langle \eta | v_{x,a} \otimes w_{y,b} \rangle |^2.\]

 Below we record a few facts about the elements of $C_{vect}^s$ and $C_{q_1}^s$.
 
 \begin{thm} Let $p(a,b|x,y) \ge 0, \forall a,b,x,y$.  Then
 \begin{enumerate}
 \item $p \in C_{vect}^s(n,n,k,k)$ if and only if there is a real Hilbert space $\cl H$, a unit vector $\eta \in \cl H$ and for each $1 \le x \le n$ a set of orthogonal vectors $\{ v_{x,a}: 1 \le a \le k \}$ such that $\sum_a v_{x,a} = \eta$ and $p(a,b|x,y) = \langle v_{x,a} | v_{y,b} \rangle$,
 \item  if $p(a,b|x,y)$ is an extreme point of the convex hull of $C_{q_1}^s(n,n,k,k)$, then there is a complex Hilbert space $\cl H$ of dimension $d \le k$ and sets of orthogonal vectors $\{ v_{x,a} : 1 \le a \le k \}$ in $\cl H$ such that $d$ of the vectors are unit norm and the remaining $k-d$ vectors are 0
 with $p(a,b|x,y) =\frac{1}{d} | \langle v_{x,a} | v_{y,b} \rangle |^2.$
 \end{enumerate}
 \end{thm}
 \begin{proof} The first statement follows from the observation that $0=p(a,b|x,x) = \langle v_{x,a} | w_{x,b} \rangle, a \ne b$ and $1= \sum_{a,b} p(a,b|x,x) = \sum_a p(a,a|x,x)$ implies that $v_{x,a} = w_{x,a}$.

 To see the second statement let $p(a,b|x,y)= \langle E_{x,a} \otimes F_{y,b} \eta | \eta \rangle$ be any element of $C^s_{q_1}$, so that the rank of each $E_{x,a}$ is at most 1. By the results of \cite{PSSTW}, if we let $\cl A \subseteq B(\cl H_A)$ denote the algebra generated by by Alice's operators then the functional $\tau: \cl A \to \bb C$ defined by
 \[ \tau(X) = \langle X \otimes I \eta | \eta \rangle,\]
 is a trace on that algebra and $p(a,b|x,y) = \tau(E_{x,a}E_{y,b})$.  Since $\cl H_A$ is finite dimensional, $\cl A$ must be *-isomorphic to a direct sum of matrix algebras and the image of $E_{x,a}$ in each component of this direct sum decomposition is of rank at most 1. On each matrix algebra, there is a unique trace, namely, the ordinary trace, normalized so that the trace of the identity is 1. Thus, $\tau$ is a convex combination of traces on matrix algebras and this direct sum decomposition expresses $p$ as a convex combination of elements of $C^s_{q_1}$. Since $p$ is an extreme point, either there is only one summand or each summand gives rise to a multiple of $p$.
 The result now follows from the observation that since $p \in C_{q_1}$, the restriction of $E_{x,a}$ to each summand is still a projection of rank 0 or 1 and the fact that
 \[\Tr(v_{x,a}v_{x,a}^* v_{y,b} v_{y,b}^*) = |\langle v_{x,a} | v_{y,b} \rangle |^2.\qedhere\]
 \end{proof}

 \begin{cor}\label{cor:q1} Let $\cl G=(X, X, A, A, V)$ be a game and let $\pi$ be a density on inputs.  Then
 \[ \omega^s_{q_1}(\cl G, \pi) = \sup \left\{\frac{1}{d} \sum_{(x,y,a,b) \in W} \pi(x,y)| \langle v_{x,a} | v_{y,b} \rangle |^2\right\},\]
where the supremum is over all $0\le d \le k$ and sets of orthogonal vectors $\{ v_{x,a} : 1 \le a \le k \}$ in $\bb C^d$ such that $d$ of the vectors are unit norm and the remaining $k-d$ vectors are 0.
\end{cor}

The corollary follows from the fact that $\bb E_{\cl G, \pi}(p)$ is a convex function of $p$.

Next we turn our attention to a new characterization of the synchronous qc-value of a game that does not appear in \cite{HMNPR} but will be useful for our calculations.   To this end assume that $\cl G$ is a game with $X=Y,\,\, A = B$, with a verifier $V$, that is not necessarily a synchronous game. Let $n=|X|, \, k = |A|$ and let $\bb F(n,k)$ denote the free product of $n$ copies of the cyclic group of order $k$. This group has generators $\{ u_x: x \in X \}$ with $u_x^k = 1$, where $1$ denotes the identity element.  We let $\bb C(\bb F(n,k))$ denote the group algebra of this group, which we regard as consisting of finite linear combinations of elements of the group (rather than finitely supported functions on the group) and let $C^*(\bb F(n,k))$ denote the full group C*-algebra.

If we let $\xi := e^{2 \pi i/k}$, then the elements of $\bb C(\bb F(n,k))$ defined by
\[ e_{x,a} := \frac{1}{k}\sum_{j=0}^{k-1} (\xi^{a-1}u_x)^j,\]
are the spectral projections of the unitary $u_x$ and these elements also generate $\bb C(\bb F(n,k))$.

Assume that $\cl G$ also has a prior density $\pi$ on inputs.  Following \cite{HMNPR} we set
\[ R_{\cl G, \pi} = \sum_{(x,y,a,b) \in W} \pi(x,y) e_{x,a} e_{y,b} e_{x,a}.\]

We recall the following theorem from \cite{HMNPR}.

\begin{thm}\cite{HMNPR} Let $\cl G= (X, X, A, A, V)$ be a game and let $\pi$ be a density on inputs.  Then
\[ \omega_{qc}^s(\cl G, \pi) = \sup \{ \tau(R_{\cl G, \pi}) \vert \,\, \tau: C^*(\bb F(n,k)) \to \bb C \text{ is a tracial state } \}.\]
\end{thm}

We let $SOS \subseteq \bb C(\bb F(n,k))$ denote the cone {\it sums of squares}, i.e., the set of all elements that can be written as a sum of elements of the form $w^*w$ with $w \in \bb C(\bb F(n,k))$.  We let
$\cl C \subseteq \bb C(\bb F(n,k))$ denote the space of {\it commutators}, i.e., the linear span of all elements of the form $wz- zw, \,\, w,z \in \bb C(\bb F(n,k))$.

While the above theorem gives lower bounds on the synchronous value, the following theorem is useful for giving upper bounds on the synchronous value. It is closely related to the results on trace positive elements found in \cite{JP}.

\begin{thm}\label{thm:SOS+C} Let $\cl G= (X, X, A, A, V)$ be a game and let $\pi$ be a density on inputs. Then
\[ \omega^s_{qc}(\cl G, \pi) = \inf \{ t \vert \,\, t1 - R_{\cl G, \pi} \in SOS + \cl C \}.\]
\end{thm}
\begin{proof} Let $t > \omega_{qc}^s(\cl G, \pi)$ and choose $\epsilon >0$ so that  $\tau(t1 - R_{\cl G, \pi}) - \epsilon >0$ for every tracial state on $C^*(\bb F(n,k))$. By a theorem of Cuntz and Pedersen~\cite[Theorem~2.9]{CP}, this implies that $(t- \epsilon)1 - R_{\cl G, \pi}$ is in the closure of $C^*(\bb F(n,k))^+ + \cl C$. From this it follows that we may choose some element $w \in \cl C$ so that $t1 - R_{\cl G, \pi} - w \in C^*(\bb F(n,k))^+ \cap \bb C(\bb F(n,k))$.  By a theorem of Ozawa~\cite[Theorem~1]{Oz14}, this guarantees that $t1 - R_{\cl G, \pi} - w \in SOS$.  Thus, $t1 - R_{\cl G, \pi} \in SOS + \cl C$.

Similarly, if $t < \omega_{qc}^s(\cl G, \pi)$, then by the above theorem, there is a tracial state $\tau$ so that
$\tau(t1 - R_{\cl G, \pi}) <0$ and this guarantees that $t1 - R_{\cl G, \pi} \notin SOS + \cl C$.
\end{proof}

\begin{defn} Given a game $\cl G=(X, X, A, A, V)$ with density $\pi$, we define the {\bf cost matrix} of the game to be the $(X \times A) \times (X \times A)$ matrix, denoted $C_{\cl G, \pi}$, with entries,
  \[ c_{(x,a), (y,b)} =  \frac{1}{2}\big( \pi(x,y)V(x,y,a,b) + \pi(y,x) V(y,x,b,a) \big).\]
Let $|x,a\rangle\langle y,b|$ denote a typical $(X\times A)\times (X\times A)$ matrix unit.  The set of {\bf cost-free matrices} $\cl Z(\cl G)$ is the real linear subspace spanned by the symmetric matrices $|x,a\rangle\langle x,b|+|x,b\rangle\langle x,a|$ for $x\in X$, $a\ne b\in A$, and the diagonal matrices $D_{x,y}=\sum_{a\in A}|x,a\rangle\langle x,a|-|y,a\rangle\langle y,a|$ for $x\ne y\in X$.
We set
\[ H_{\cl G, \pi} = \sum_{a,b,x,y} c_{(x,a),(y,b)}(e_{x,a}e_{y,b} + e_{y,b} e_{x,a}),\]
viewed as an element of the corresponding group algebra.
Note that $C_{\cl G, \pi}^t = C_{\cl G, \pi}$ and that $H_{\cl G, \pi} = H_{\cl G, \pi}^*$.
\end{defn}

 \begin{cor}\label{cor:cost-matrix} Given a game $\cl G=(X, X, A, A, V)$ with density $\pi$,
 \[ \omega^s_{qc}(\cl G, \pi) = \sup \{ \tau( H_{\cl G, \pi}) \} \le |X|\inf_{Z\in \cl Z(\cl G)} \lambda_{max}(C_{\cl G, \pi}+Z),\]
 where $\tau$ is an arbitrary tracial state and $\lambda_{max}$ denotes the maximal eigenvalue.%
 \end{cor}
 \begin{proof} The first equality follows from the fact that $\tau(H_{\cl G, \pi}) = \tau(R_{\cl G, \pi})$ for any trace.
 If we choose some ordering on the set $X \times A$ and let $E$ denote the column vector with entries the elements $e_{x,a}$ then for $Z\in \cl Z(\cl G)$, we have that $E^*ZE=0$, so
 \[ H_{\cl G, \pi} = E^*C_{\cl G,\pi}E=E^* (C_{\cl G, \pi}+Z) E.\]
 We also have $E^*E=|X|1$, so
 \[t1-H_{\cl G,\pi}=E^*(t|X|^{-1}I-C_{\cl G,\pi}-Z))E.\]
It follows that $tI - H_{\cl G, \pi} \ge 0$ in the corresponding group C*-algebra for
 $t > |X|\lambda_{max}(C_{\cl G, \pi}+Z).$
 \end{proof}

\section{Unique games}

For a unique game $\cl G$, it is known~\cite{KRT} that if $\cl G$ has a perfect $vect$-density, then $\cl G$ has a perfect $q$-density.  We strengthen this by showing that the same applies if we replace $q$-densities with deterministic strategies.

\begin{thm} Let $\cl G=(X,Y,A,f)$ be a unique game.
  If $\cl G$ has a perfect vect-density, then $\cl G$ has a
  perfect deterministic strategy. If $X=Y$ and $\cl G$ has a perfect synchronous vect-density, then $\cl G$ has  a perfect synchronous deterministic strategy.
 \end{thm}
 \begin{proof} Let $\{ v_{x,a}: x \in X, a \in A \}$ and $\{ w_{y,b} : y \in Y, b \in A \}$ be sets of vectors in a Hilbert space $\cl H$ satisfying:
 \begin{itemize}
 \item $v_{x,a} \perp v_{x,a'}, \forall x \in X, a \ne a',$
 \item $w_{y,b} \perp w_{y,b'}, \forall y \in Y, b \ne b'$,
 \item $\exists \| \eta \|=1$ such that $\sum_a v_{x,a} = \sum_b w_{y,b} = \eta,  \forall x,y \in X$
 \end{itemize}
 such that the density given by
 \[ p(a,b | x,y) := \langle v_{x,a} \vert w_{y,b} \rangle,\]
 is a perfect density for $\cl G$.

 This implies that whenever $\langle v_{x,a} \vert w_{y,b} \rangle \ne0$, we have $a = f(x,y)b$.

 Fix any $x_0\in X$ and $y_0\in Y$. Choose $b_0\in A$
 such that $w_{y_0,b_0} \ne 0$ and define $h_A:X \to A$ by
 \[ h_A(x) = f(x,y_0)b_0.\]
 Now set $a_0 = h_A(x_0)$ and define $h_B: Y \to A$ by
 \[ h_B(y) = f(x_0,y)^{-1}a_0.\]
 Note that $h_B(y_0)= b_0$.
 We claim that this pair of functions defines a perfect deterministic strategy.

  Since
 \[ 0 \ne \|w_{y_0,b_0}\|^2 = \langle w_{y_0,b_0} \vert \eta \rangle = \sum_a \langle w_{y_0,b_0} \vert v_{x,a} \rangle = \langle w_{y_0, b_0} \vert v_{x, h_A(x)} \rangle,\]
 we see that $v_{x, h_A(x)} \ne 0, \forall x$ and so in particular $v_{x_0, a_0} \ne 0.$
 We also have that
 \[ \|v_{x, h_A(x)} \|^2 = \sum_b \langle w_{y_0,b} \vert v_{x, h_A(x)} \rangle = \langle w_{y_0,b_0} \vert v_{x, h_A(x)} \rangle.\]
 These two equations together imply that $v_{x,h_A(x)} = w_{y_0, b_0}$.

 Similarly, we have that $w_{y, h_B(y)} = v_{x_0, a_0}$.  Hence,
 \[ p(h_A(x),h_B(y) | x,y) = \|w_{y_0, b_0} \|^2 \ne 0,\]
 and since $p$ is a perfect density for $\cl G$, we have $h_A(x)=f(x,y)h_B(y)$. So the functions $h_A,h_B$ define a perfect deterministic strategy. Moreover, if $X=Y$ and $p$ is synchronous, then it follows immediately that $h_A=h_B$, and we have a perfect deterministic strategy.
 \end{proof}

We now give an approximate version of the previous result.

 \begin{thm}\label{near1}
   Let $\mathcal{G}=(X,Y,A,f)$ be a unique game with $k=|A|\ge2$. Let
   $\pi$ be a density on inputs such that
   $\pi_0:=\min_{x,y}\pi(x,y)>0$ and let
   $\epsilon>0$ be sufficiently small that
   $\epsilon+2k(\epsilon/\pi_0+\sqrt{\epsilon/\pi_0})<1$. %
     If $\omega_{vect}(\mathcal{G},\pi)\ge 1- \epsilon$, then $\omega_{vect}(\cl G, \pi) = \omega(\cl G, \pi) =1$.
   If $(\mathcal{G},\pi)$ has no perfect deterministic strategy, then
   we have $\omega_{vect}(\cl G, \pi) <1-\epsilon$.
 \end{thm}
 \begin{proof}
   Suppose that $\omega_{vect}(\mathcal{G},\pi)=\bb E_{\cl G,\pi}(p)\ge 1-\epsilon$ where $p(a,b|x,y)=\langle v_{x,a}|w_{y,b}\rangle$ is a vect-density. We will show that $\omega_{det}(\mathcal{G},\pi)=1$.

   Let $W$ be the winning set and define
\[ r= \max \{ p(a,b|x,y) :  (x,y,a,b) \in W \},\]
and let $(x_0,y_0,a_0,b_0) \in W$ satisfy
\[r= p(a_0,b_0 | x_0, y_0) = \langle v_{x_0,a_0} | w_{y_0,b_0} \rangle.\]
Since $\mathcal{G}$ is a unique game, we have that
\[ 1 - \epsilon \le \bb E_{\cl G,\pi}(p) \le \sum_{(x,y,a,b) \in W} \pi(x,y) r = k r.\]
Define $h_A:X\to A$, $h_B:Y\to A$ using $(x_0,y_0,a_0,b_0)$, exactly as in the previous proof. Then for $\eta=\sum_b w_{y,b}=\sum_a v_{x,a}$, we have
 \[0\le \|w_{y_0,b_0}\|^2=\langle w_{y_0,b_0}|\eta\rangle=\langle w_{y_0,b_0}|v_{x,h_A(x)}\rangle+ \delta_{b_0}^A(x)\]
where $\delta_{b_0}^A(x)=\sum_{a\ne h_A(x)}p(a,b_0|x,y_0)$. And
\[ \|v_{x,h_A(x)}\|^2=\langle v_{x,h_A(x)}|\eta\rangle=\langle v_{x,h_A(x)}|w_{y_0,b_0}\rangle+\delta_{b_0}^B(x)\]
where
$\delta_{b_0}^B(x)=\sum_{b\ne b_0}p(h_A(x),b|x,y_0)$.

If vectors $v,w$ in a Hilbert space satisfy $|\|w\|^2-\langle v,w\rangle|\leq \delta_1$ and $|\|v\|^2-\langle w,v\rangle|\leq \delta_2$, then $\|w-v\|^2\leq \delta_1+\delta_2$.
So
\[ \|v_{x,h_A(x)}-w_{y_0,b_0}\|^2\leq \delta_{b_0}(x):= \delta_{b_0}^A(x)+\delta_{b_0}^B(x).\]
Similarly,
\[ \|v_{x_0,a_0}-w_{y,h_B(y)}\|^2\leq \delta_{a_0}(y):=\delta_{a_0}^A(y)+\delta_{a_0}^B(y)\]
where
$\delta_{a_0}^A(y)=\sum_{a\ne a_0}p(a,h_B(y)|x_0,y)$
and
$\delta_{a_0}^B(y)=\sum_{b\ne h_B(y)}p(a_0,b|x_0,y)$.
So
\begin{align*}
  p(h_A(x),h_B(y)|x,y)
  &
  =\langle v_{x,h_A(x)}|w_{y,h_B(y)}\rangle
  \\&=\langle w_{y_0,b_0}+(v_{x,h_A(x)}-w_{y_0,b_0})|v_{x_0,a_0}+(w_{y,h_B(y)}-v_{x_0,a_0})\rangle
  \\&=r +\langle v_{x,h_A(x)}-w_{y_0,b_0}|v_{x_0,a_0}\rangle+\langle w_{y_0,b_0}|w_{y,h_B(y)}-v_{x_0,a_0}\rangle\\&\hphantom{{}=r{}}+\langle v_{x,h_A(x)}-w_{y_0,b_0}|w_{y,h_B(y)}-v_{x_0,a_0}\rangle
  \\
  &\ge  \frac{1-\epsilon}{k}-\gamma(x,y)
\end{align*}
where
\begin{align*}
  \gamma(x,y)
  &=|\langle v_{x,h_A(x)}-w_{y_0,b_0}|v_{x_0,a_0}\rangle|+|\langle w_{y_0,b_0}|w_{y,h_B(y)}-v_{x_0,a_0}\rangle|\\&\hphantom{{}={}}+|\langle v_{x,h_A(x)}-w_{y_0,b_0}|w_{y,h_B(y)}-v_{x_0,a_0}\rangle|\\
  &\le \sqrt{\delta_{b_0}(x)}\|v_{x_0,a_0}\|+\sqrt{\delta_{a_0}(y)}\|w_{y_0,b_0}\|+\sqrt{\delta_{b_0}(x)\delta_{a_0}(y)}
  \\
  &\leq 2\sqrt{\frac{\epsilon}{\pi_0}}+\frac{\epsilon}{\pi_0}%
\end{align*}
where we have made the fairly crude estimates
\begin{align*}\delta_{b_0}(x,y_0)&\le \sum_{a,b:V(x,y_0,a,b)=0}p(a,b|x,y_0)\\&\leq \frac1{\pi(x,y_0)}\sum_{x,y,a,b:V(x,y,a,b)=0}\pi(x,y)p(a,b|x,y)=\frac{\epsilon}{\pi(x,y_0)}\leq \frac{\epsilon}{\pi_0}
\end{align*}
and similarly, $\delta_{a_0}(y)\le \frac{\epsilon}{\pi_0}$. %
This all implies that
\begin{equation*}
  p(h_A(x),h_B(y)|x,y)-\frac\epsilon{\pi_0} \ge \frac1k\left(1-\left(\epsilon+2k\left(\frac\epsilon{\pi_0}+\sqrt{\frac{\epsilon}{\pi_0}}\right)\right)\right)>0
\end{equation*}
so $\pi_0 p(h_A(x),h_B(y)|x,y)>\epsilon$. Hence,
the expected value $\bb E_{\cl G,\pi}(h_A,h_B)$ of the deterministic strategy $(h_A,h_B)$ satisfies
\begin{align*}
  (1-\bb E_{\cl G,\pi}(h_A,h_B))\epsilon&=\sum_{x,y:V(x,y,h_A(x),h_B(y))=0}\pi(x,y)\epsilon
  \\&
  < \pi_0\sum_{x,y:V(x,y,h_A(x),h_B(y))=0}\pi(x,y) p(h_A(x),h_B(y)|x,y)
  \\&
  \leq \pi_0\sum_{x,y,a,b:V(x,y,a,b)=0}\pi(x,y) p(a,b|x,y)
  \\&=\pi_0(1-\bb E_{\cl G,\pi}(p))\leq\pi_0 \epsilon.
\end{align*}
So
\[ \bb E_{\cl G,\pi}(h_A,h_B)> 1-\pi_0.\]
Since $\bb E_{\cl G,\pi}(h_A,h_B)$ is a sum of values of $\pi$, this implies that $\bb E_{\cl G,\pi}(h_A,h_B)=1$, hence $\omega_{det}(\cl G,\pi)=1$.
\end{proof}
\begin{remark}\label{remark:constant}
  For $k\ge2$, any positive $\epsilon\le \frac{\pi_0}{(1+\sqrt{1+3/k})^2k^2}$ will be sufficiently small for the above hypothesis to hold. In particular, since $(1+\sqrt{1+3/k})^2\le (1+\sqrt{5/2})^2<7$, we have that $\epsilon=\frac{\pi_0}{7k^2}$ will do.
  Moreover, if we know the game's vect-value is achieved by a uniform vector strategy, meaning $\|v_{x,a}\|=\|w_{y,b}\|=\tfrac1{\sqrt k}$ for every $x,y,a,b$, then this gives us an improved upper bound $\gamma(x,y)\le 2\sqrt{\epsilon/\pi_0k}+\epsilon/\pi_0$ in the proof. This lets us relax the hypothesis to $\epsilon+ 2k\epsilon/\pi_0+2\sqrt{k\epsilon/\pi_0}<1$, which is satisfied by, for example, $\epsilon=\frac{\pi_0}{9k}$.
\end{remark}

\section{A groups analogue of XOR games}
 Given a finite group $\Omega$ and a set $X$, let $f:X \times Y \to \Omega$ be a function.
 We define a two person game, denoted by $\cl G= (X, Y, \Omega, f)$, by setting the value function to be
 \[ V(x,y,a,b) =1 \iff a b^{-1} = f(x,y).\] In the case that $X=Y$, we simplify this notation further to $\cl G= (X, \Omega, f)$.
 We refer to any game given as above as a {\bf group based game}.

 When $\Omega = (\bb Z_2, +)$, then since $a-b=a+b$, this becomes the usual definition of an XOR game.
 Note that a deterministic strategy for such a game is a pair of functions $h_A: X \to \Omega$ and $h_B:Y \to \Omega$. Such a strategy is perfect if and only if $h_A(x) h_B(y)^{-1} = f(x,y), \forall x,y$.

 In the case that $X=Y$, the function $f$ defines a {\it synchronous game} if and only if
 $f(x,x) =e, \forall x$, where $e$ denotes the identity of $\Omega$. We call the game {\it symmetric} provided that $f(y,x) = f(x,y)^{-1}, \forall x,y \in X$.

 Note that a pair of functions defines a perfect strategy for a synchronous game $(X,\Omega,f)$ if and only if $h_A = h_B =h$ with $h(x) h(y)^{-1} = f(x,y), \forall x,y \in X$. Thus, for a synchronous group based game to have a perfect strategy it must also be symmetric.

We also remark that since the groups that we are interested in are finite, $\bb C(\Omega) = C^*(\Omega)$.

\begin{prop}[Averaging]\label{averaging} Let $\cl G=(X, Y, \Omega, f)$ be a group based game and let $\pi$ be a density on inputs.  If $p \in C_t$, respectively, $p \in C^s_t$, for $t=loc, q, qs,qa,qc, vect, ns$ and we set
\[ \hat{p}(a, b | x, y) := \frac{1}{|\Omega|} \sum_{c \in \Omega} p(ac, bc | x, y),\]
then $\hat{p} \in C_t$, respectively, $\hat{p} \in C^s_t$ with $\hat{p}(a,b|x,y) = \hat{p}(ac,bc|x,y), \forall c\in \Omega$ and
\[ \bb E_{\cl G, \pi}(p) = \bb E_{\cl G, \pi}(\hat{p}).\]
\end{prop}

We leave the proof to the reader.

Recall that given a C*-algebra $\cl A$ and a set $X$ a function
\[ g: X \times X \to \cl A\]
is called {\it positive definite} provided that for every finite set $\{ x_1,...,x_n \} \subseteq X$ we have that
\[ (g(x_i, x_j)) \in M_n(\cl A)^+.\]
Given a finite group $\Omega$ we let $\tau_{\Omega}: \bb C(\Omega) \to \bb C$ denote the canonical trace that gives the coefficient of the group identity.

\begin{prop} Let $X$ be a finite set, $\Omega$ a finite group and let $p \in C_{qc}^s(X, X, \Omega, \Omega)$. Then there is a C*-algebra $\cl A$ with tracial state $\tau$ and a positive definite function $h: X \times X \to C^*(\Omega) \otimes \cl A$
such that
\[ \hat{p}(a,b|x,y) = \tau_{\Omega} \otimes \tau((a^{-1} \otimes 1)h(x,y)(b \otimes 1)).\]
For a group based game $(X,\Omega, f)$ and an input density $\pi:X\times X\to [0,1]$, we have
\[ \bb E_{\cl G,\pi}(p)=|\Omega|\sum_{x,y}\pi(x,y)\tau_\Omega\otimes \tau\big((f(x,y)^{-1}\otimes 1)h(x,y)\big).\]
\end{prop}
\begin{proof} We have a tracial C*-algebra $(\cl A, \tau)$ and an $(n,k)$-PVM $\{ e_{x,a} \}$ in $\cl A$ such that
\[ p(a,b|x,y) = \tau(e_{x,a} e_{y,b}).\]
Set
\[ w(x) = \frac{1}{\sqrt{|\Omega|}}\sum_{a \in \Omega} a \otimes e_{x,a} \in C^*(\Omega) \otimes \cl A.\]
Then $h(x,y) =w(x) w(y)^*$ is a positive definite function into $C^*(\Omega) \otimes \cl A$.

We compute
\begin{align*} \tau&_{\Omega} \otimes \tau ((a^{-1} \otimes 1) h(x,y) (b \otimes 1)) = \frac{1}{|\Omega|}(\tau_{\Omega} \otimes \tau) \big( \sum_{c,d \in \Omega} a^{-1}cd^{-1}b \otimes e_{x,c}e_{y,d} \big) \\ &= \frac{1}{|\Omega|}\sum_{c,d:a^{-1}c=b^{-1}d} p(c,d|x,y) =
\frac{1}{|\Omega|}\sum_{r\in \Omega} p(ar, br | x,y) = \hat{p}(a,b | x,y).
\end{align*}
If $b,c\in \Omega$ and $z\in C^*(\Omega)\otimes \cl A$, then \[\tau_\Omega\otimes \tau(((cb)^{-1}\otimes 1)z(b\otimes 1))=\tau_\Omega\otimes \tau((c^{-1}\otimes 1)z).\] Hence,
\begin{align*}
  \bb E_{\cl G,\pi}(p)&=\bb E_{\cl G,\pi}(\hat p)\\&=\sum_{x,y}\pi(x,y)\sum_{b\in \Omega,a=f(x,y)b} \tau_\Omega\otimes \tau((a^{-1}\otimes 1)h(x,y)(b\otimes 1))\\
  &=|\Omega|\sum_{x,y}\pi(x,y)\tau_\Omega\otimes\tau((f(x,y)^{-1}\otimes 1)h(x,y)).\qedhere
\end{align*}
\end{proof}
\begin{remark} In the case that $p \in C^s_q(X, X, \Omega, \Omega)$, one can assume that in addition $\cl A$ is a finite dimensional C*-algebra. In the case that $p \in C^s_{qa}(X, X, \Omega, \Omega)$, one can assume that $\cl A$ is an ultrapower of the hyperfinite $II_1$-factor and that $\tau$ is the canonical trace. This follows from results of \cite{KPS}, see \cite{HP} for further explanation. Also, by results of \cite{KPS}, $C^s_q= C^s_{qs}$.
\end{remark}

\subsection{Perfect strategies for group based games}
\subsubsection{Deterministic and vector strategies}

 \begin{prop} Let $\cl G= (X, \Omega, f)$ be a synchronous group based game.  Then $\cl G$ has a perfect vect-strategy if and only if the map $f:X \times X \to C^*(\Omega)$ is positive definite.
 \end{prop}
  \begin{proof} Since $\cl G$ is a unique game, it will have a perfect vect-strategy if and only if it has a perfect deterministic strategy.
 If we consider the case that $X= [n]$ then a perfect deterministic strategy for~$\cl G$ corresponds to factoring
 \[ (f(i,j)) = (h(1),....,h(n))^t(h(1)^*,...,h(n)^*),\]
 which implies that $(f(i,j)) \in M_n(C^*(\Omega))^+.$ Hence, $f$ is positive definite.

 Conversely, if $f$ is positive definite, then $(f(i,j)) \in M_n(C^*(\Omega))^+.$ But, since each $f(i,j)$ is a unitary in $C^*(\Omega)$, this positivity implies that
 $f(i,j) = f(1,i)^* f(1,j)$ and so $h(j) = f(1,j)$ is a perfect deterministic strategy.
 \end{proof}

\begin{thm}
   Let $\mathcal{G}=(X, Y, \Omega, f)$ be a group based game with input density $\pi$, assume
   $\pi_0=\min_{x,y}\pi(x,y)>0$ and set $\epsilon= \frac{\pi_0}{7|\Omega|}$. If $\omega_{vect}(\mathcal{G},\pi)\ge 1- \epsilon$, then $\omega_{vect}(\cl G, \pi) = \omega(\cl G, \pi) =1$. If $\mathcal{G}$ has no perfect deterministic strategy, then
   $\omega_{vect}(\cl G, \pi) <1-\epsilon$.
 \end{thm}
\begin{proof} Let $p(a,b | x,y) = \langle v_{x,a} \vert w_{y,b} \rangle$ be the vect-correlation that attains the value.
If we set
\[ \hat{v}_{x,a} = \frac{1}{\sqrt{\Omega}} \bigoplus_{c \in \Omega} v_{x,ac}, \,\, \hat{w}_{y,b} = \frac{1}{\sqrt{\Omega}} \bigoplus_{c \in \Omega} w_{y, bc},\]
then
\[ \hat{p}(a, b | x, y) = \langle \hat{v}_{x,a} \vert \hat{w}_{y,b} \rangle,\]
is a new vect-correlation that also attains the value.  Moreover, these vectors have the property that
\[ \| \hat{v}_{x,a} \|^2 = \frac{1}{|\Omega|}\sum_{c \in \Omega} \|v_{x,c}\|^2 = \frac{1}{|\Omega|}= \|\hat{w}_{y,b} \|^2.\]
Repeating the proof of Theorem~\ref{near1} as outlined in Remark~\ref{remark:constant} yields the result.
\end{proof}

  \section{Labelling Directed Graphs and Synchronous Games Based on the Cyclic Group of Order 3}

In this section we show that symmetric synchronous games based on the cyclic group of order 3 correspond to a certain labelling problem for directed graphs.  We first recall the relationship between the cut number of undirected graphs and synchronous XOR games.

  Given an undirected simple graph $G=(V,E)$ the cut number $\Cut(G)$ is the maximal number of (undirected) edges in a two-colorable subgraph of $G$.
  We regard the edge set $E$ as a subset of $V\times V$, where $(x,y)\in E\iff (y,x)\in E$, and $(x,x) \notin E$, so that
  \[\Cut(G)=\frac12\max_{w:V\to \bb Z_2}\big|\{(x,y)\in E: w(x)\ne w(y)\}\big|.\]
  We define $f:V \times V \to \bb Z_2$ by
  \[ f(x,y) = \begin{cases} 1 & (x,y) \in E, \\ 0 & (x,y) \notin E. \end{cases}\]
  The graph can be 2-colored if and only if there is a function $w:V \to \bb Z_2$ such that $w(x) - w(y) = 1, \forall (x,y) \in E$.
  Note that generally, if $(x,y) \notin E$ then it is not necessary that $w(x) - w(y) =0$. For example the 6-cycle can be two-colored, but some pairs of non-adjacent vertices will receive opposite colors. This problem is fixed by choice of density on the input set.  If we let our density be
  \[ \pi_E(x,y) = \frac{1}{|E|} \chi_E(x,y) := \frac{1}{|E|} \begin{cases} 1 & (x,y) \in E \\ 0 & (x,y) \notin E \end{cases},\]
  and let $\cl G=(V,\bb Z_2,f)$ denote the corresponding group based game, then
  \[\omega^s_{loc} ( \cl G,\pi_E) = \frac{1}{|E|} \max_{w:V \to \bb Z_2} \big|\{(x,y)\in E |w(x)-w(y)=1\}\big|.\]
  Hence,
  \[ 2 \Cut(G) = |E|\omega^s_{loc}(\cl G,\pi_E).\]
  This was used in \cite{HMNPR} to define a $t$-quantum cut number for $t=q,qc$ as
 \[  \Cut_t(G) = \frac{|E|}2 \omega^s_t(\cl G,\pi_E).\]

  We now wish to do a similar construction for directed graphs.

  By a {\it directed graph} $G=(X,D)$ on a vertex set $X$ we mean a subset $D \subseteq X \times X$ such that $(x,x) \notin D$ and $(x,y) \in D \implies (y,x) \notin D$. Such a graph is sometimes called a simple, antisymmetric directed graph, and can be pictured as a collection of arrows from vertex $x$ to $y$ if and only if $(x,y) \in D$. We let $D^t = \{ (y,x): (x,y) \in D \}$ and set $E= D \cup D^t$,  so that $(X, E)$ is the graph on $X$ that one obtains by removing the arrows but keeping the edges.
  Let $\bb Z_3$ denote the group $\{0,1,2\}$ under addition modulo
  $3$. We can identify a directed graph $G=(X, D)$ with a function
  \[ f: X \times X \to \bb Z_3 \text{ such that } f(x,y) = \begin{cases} 0, & (x,y),(y,x) \notin D,\\ 2, & (x,y) \in D, \\ 1, & (x,y) \in D^t, \end{cases} \]
  so that $f(x,y) = f(y,x) =0$ represents no edge between $x$ and $y$.

  Conversely, given any function $f:X \times X \to\bb Z_3$ with $f(y,x) = -f(x,y)$ and $f(x,x) =0$ we can associate a directed graph by setting $D = \{ (x,y):  f(x,y) = 1 \}$.

  In this manner we see that every synchronous, symmetric group based game of the form $\cl G= (X, \bb Z_3, f)$ can be identified with a directed graph $G=(X, D)$.

  A {\bf vertex 3-labelling} of the directed graph $G$ is a function $w:X \to \bb Z_3$ such that
  \[ w(x)- w(y) = f(x,y), \forall (x,y) \in E.\]
  Equivalently, $w(y)=w(x)+1$ for all $(x,y)\in D$, or equivalently, $w(y)=w(x)-1$ for all $(x,y)\in D^t$.

  Thus, every synchronous group based game on $\bb Z_3$ is the vertex 3-coloring game of a directed graph, provided we take the density of the game to be supported on $E$.

  Given any directed graph on $X$, say $G=(X, D)$ represented by $f$, we take its {\bf directed cut number} to be
  \begin{align*}
    \DCut(G) &=\frac12\max_{w:X \to \bb Z_3} |   \{ (x,y) \in E : w(x)- w(y) =f(x,y) \}|\\&=\max_{w:X \to \bb Z_3} |   \{ (x,y) \in D : w(y)=w(x)+1 \}|,
  \end{align*}
  i.e., the largest number of directed edges of a vertex $3$-labelled directed subgraph $(X,D')$ of $G$, where $D'\subseteq D$.
  As above we have that
  \[ \DCut(G) = |D|\omega^s_{loc}(\cl G,\pi_E),\]
  for the corresponding \emph{directed graph game} $(\cl G,\pi_E)$, i.e., the group based game $\cl G=(X,\bb Z_3,f)$ with input density $\pi_E$ which is uniform on $E$.

 It is natural to set
 \[\DCut_t(X_d) = |D|\omega^s_t(\cl G,\pi_E)\]
 for $t=loc,q_1,q,qc,vect,ns$.

 Our next goal is to study these values. Let $\xi=e^{2\pi i/3}$.

   \begin{thm}\label{thm:value1} Let $(X, D)$ be a directed graph and $(\cl G,\pi)$ the corresponding directed graph game. Suppose $\{ e_{x,a} : x \in X, a\in \bb Z_3 \}$ are PVMs in a C*-algebra with trace $\tau$. Let $p(a,b|x,y)=\tau(e_{x,a}e_{y,b})$ for $a,b\in \bb Z_3$ and $x,y\in X$, and define unitaries $\{ U_x : x \in X \}$ by setting $U_x= e_{x,0} + \xi e_{x,1} + \xi^2 e_{x,2}$. Then
     \[ \bb E_{\cl G,\pi}(p)=\frac13+\frac2{3|D|}\sum_{(x,y)\in D}\re \tau(\xi^2 U_x^*U_y).\]
   \end{thm}
   \begin{proof}
     For $x,y\in X$ and $j\in \bb Z_3$, let \[r_j(x,y)=\sum_{a\in \bb Z_3} \tau(e_{x,a}e_{y,a+j}).\] Then $r_j(y,x)=r_{-j}(x,y)$ and $\sum_{j\in \bb Z_3}r_j(x,y)=1$. Moreover,
     \[ \tau(U_x^*U_y)=\sum_{j\in \bb Z_3}\xi^{j}r_j(x,y).\]
     It follows that
     \[ 3r_1(x,y)=1+2\re \tau(\xi^2 U_x^*U_y).\]
     For $(x,y)\in E$, we have $(x,y,a,b)\in W$ if and only if $(x,y)\in D$ and $b=a+1$ or $(y,x)\in D$ and $a=b+1$, so
     \[ W\cap (E\times \bb Z_3\times \bb Z_3)=\{(x,y,a,a+1),(y,x,a+1,a)\mid (x,y)\in D,\,a\in \bb Z_3\}.\]
     Hence,
     \begin{align*}
       \bb E_{\cl G,\pi}(p)
       &=\frac1{2|D|}\sum_{(x,y)\in D,\,a\in \bb Z_3} \tau(e_{x,a}e_{y,a+1})+\tau(e_{y,a+1}e_{x,a})\\
       &=\frac1{|D|}\sum_{(x,y)\in D,\,a\in \bb Z_3} \tau(e_{x,a}e_{y,a+1})=\frac1{|D|}\sum_{(x,y)\in D} r_1(x,y)\\
       &=\frac13+\frac2{3|D|}\sum_{(x,y)\in D}\re \tau(\xi^2 U_x^*U_y).\qedhere
     \end{align*}
   \end{proof}

   The point of the above theorem is that we can replace PVMs by unitaries of order 3 and the synchronous values of the game will be given by taking the supremum over such unitaries in all tracial C*-algebras or over all unitaries in a matrix algebra.

   If $(X,D)$ is a directed graph and $x, y\in X$, and vertices $(x_1,\dots,x_m)$ form a cycle in the undirected graph $(X,D\cup D^t)$, then the \emph{net length} of this cycle is defined to be
   \[
     |\{1\le j\le m:(x_j,x_{j+1})\in D\}|-|\{1\le j\le m:(x_{j+1},x_j)\in D\}|,
   \] where $x_{m+1}:=x_1$. If the net length of this cycle is $m$ (so
   that all of its edges are oriented in the same cyclic direction), it is
   said to be a directed $m$-cycle.
   \begin{cor}\label{cor:directed}
     Let $(X,D)$ be a directed graph with corresponding game $(\cl G,\pi)$. The following are equivalent:
     \begin{enumerate}
     \item $(\cl G,\pi)$ has a perfect synchronous $qc$-density.
     \item $(\cl G,\pi)$ has a perfect synchronous deterministic strategy.
     \item The net length of every cycle in $(X,D)$ is divisible by $3$.
     \item $(\cl G,\pi)$ has a perfect deterministic strategy.
     \end{enumerate}
     In particular, if $(X,D)$ contains a directed $k$-cycle, where $3$ does not divide $k$, then there is no perfect synchronous $qc$-density for $(\cl G,\pi)$.
   \end{cor}
   \begin{proof}
     Any density $p$ for $\cl G$ in $C_{qc}^s$ is of the form $p(a,b|x,y)=\tau(e_{x,a}e_{y,b})$ where $e_{x,a}$ and $\tau$ are as in Theorem~\ref{thm:value1}, so if $p$ is perfect for $(\cl G,\pi)$, then the unitaries $\{U_x:x\in X\}$ satisfy $\tau(\xi^2U_x^*U_y)=1$ and hence $U_x^*U_y=\xi 1$, for all $(x,y)\in D$.
     Define an equivalence relation on $X$ by $x\sim y$ if $U_x^*U_y=\xi^k1$ for some $k\in \bb Z_3$, and let $X_0$ be a complete set of equivalence class representatives. Define $w:X\to \bb Z_3$ by the condition $U_{x_0}^*U_x=\xi^{w(x)}1$, where $x\sim x_0\in X_0$. If $(x,y)\in D$, then $U_x^*U_y=\xi1$, so $x\sim x_0\sim  y$ for some $x_0\in X_0$, and \[\xi 1=U_x^*U_y=(U_{x_0}^*U_x)^*(U_{x_0}^*U_y)=\xi^{w(y)-w(x)}1,\] so $w(y)=w(x)+1$. Hence, $w$ defines a perfect synchronous deterministic strategy. So (a) and (b) are equivalent.

     A perfect synchronous deterministic strategy is the same thing as a directed graph homomorphism into the directed cycle on three vertices. By~\cite[Theorem~3.2]{BB}, the existence of such a homomorphism is equivalent to (c).

     Condition (b) trivially implies (d) so it only remains to show that (d) implies (c). Given (d), there is a perfect deterministic strategy, i.e., two functions $v,w:X\to \bb Z_3$ such that for all $(x,y)\in D$, we have $w(y)=v(x)+1$ and $v(y)=w(x)+1$. Suppose $C=(x_1,\dots,x_m)$ is a cycle in $D$ of net length~$\ell$. If $m$ is even, then a perfect synchronous deterministic strategy for this cycle is given by $c:\{x_1,\dots,x_m\}\to \bb Z_3$, $c(x_i)=v(x_i)$ if $i$ is even, and $c(x_i)=w(x_i)$ if $i$ is odd, so $\ell$ must be divisible by $3$. If $m$ is odd, suppose without loss that $(x_m,x_1)\in D$. Let $C'=(x_1,\dots,x_m,x_1',\dots,x_m')$ be the directed cycle with edges
     \[ \{ (x_m,x_1'),(x_m',x_1)\}\cup \{ (x_i,x_j),(x_i',x_j'): 1\le
       i,j\le m,\, |i-j|\le 1\}.\] So $C'$ is obtained by removing the edge
     $(x_m,x_1)$ from $C$, duplicating the resulting path, and forming a cycle by adding the edges $(x_m,x_1'),(x_m',x_1)$. The net length of $C'$ is then $2\ell$, and $C'$ has a perfect synchronous deterministic strategy $c:\{x_i,x_i':1\le i\le m\}\to \bb Z_3$, given by
     \[ c(x_i)=
       \begin{cases}
         v(x_i)&\text{$i$ even},\\
         w(x_i)&\text{$i$ odd},
       \end{cases}\qquad
       c(x_i')=
       \begin{cases}
         w(x_i)&\text{$i$ even},\\
         v(x_i)&\text{$i$ odd}.
       \end{cases}\]
     Hence, $3$ divides $2\ell$, so $3$ divides $\ell$.
   \end{proof}
   \begin{remark}
     The game $(\cl G,\pi)$ corresponding to a directed graph $(X,D)$ has input density which is zero for many inputs, so the last result does not follow immediately from Theorem~\ref{near1}.
   \end{remark}

\subsection{The directed $k$-cycle}
Let $k\ge 4$. Let $C_k^d=(X,D)$ be the directed $k$-cycle where $X=\bb Z_k=\{0,1,\dots,k-1\}$, equipped with addition modulo~$k$, and $D=\{(x,x+1):x\in X\}$. We write $\omega_t(C_k^d)$ and $\omega_t^s(C_k^d)$ for the $t$-value (respectively, synchronous $t$-value) of the corresponding directed graph game.

We wish to estimate some of these values. Note if $k$ is divisible by $3$ then by Corollary~\ref{cor:directed} this game has a perfect synchronous deterministic strategy, so these values are all~$1$. So we will always suppose that $k$ is not divisible by~$3$.

\begin{prop}
  For any $k\ge4$ which is not divisible by $3$, we have
  \[
    \omega_{loc}(C_k^d)=\begin{cases}1-1/k&\text{if $k$ even,}\\
      1-1/2k&\text{if $k$ is odd}
    \end{cases}
  \]
  and
  \[\omega_{loc}^s(C_k^d)=1-1/k.\]
\end{prop}
\begin{proof}
  Corollary~\ref{cor:directed} shows that none of these values
  is~$1$. Hence, the integer $\DCut(C_k^d)$ is strictly less than $k$, the number of directed edges, so we have $\omega_{loc}^s(C_k^d)\le \frac{k-1}k=1-1/k$, and this is attained by the vertex $3$-coloring \[w:X\to \bb Z_3,\; w(x)=x\bmod 3.\]
  Similarly, we have $\omega_{loc}(C_k^d)\le 1-1/2k$ and if $k$ is odd, this is attained by $v,w:X\to \bb Z_3$,
  \[ v(x)=x\bmod 3,\;w(x)=x+k\bmod 3.\]
  If $k$ is even and $w_0,w_1:X\to \bb Z_3$, then the expected value of the deterministic density corresponding to $(w_0,w_1)$ is the mean of the two synchronous deterministic densities arising from $g,h:X\to \bb Z_3$,
  \[ g(x)=w_{x\bmod 2}(x),\; h(x)=w_{x+1\bmod 2}(x).\]
  Hence, $\omega_{loc}(C_k^d)=\omega_{loc}^s(C_k^d)$ in this case.
\end{proof}

Let $\theta_k=\frac{2\pi}{3k}$ and $c_k=\cos\theta_k$.
\begin{prop}
  We have
  \[ \left(\frac{1+2c_k}3\right)^2 \le \omega_{q_1}^s(C_k^d)\le \omega_q^s(C_k^d)\le \omega_{qc}^s(C_k^d)\le \frac{1+2c_k}3\]
  and
  \[ \frac{1+2c_k^2}3\le \omega_{q_1}(C_k^d).\]
\end{prop}
\begin{proof}
  Let $U$ be the $3\times 3$ matrix of rotation by $\theta_k=\frac{2\pi}{3k}$ about the axis $(1,1,1)^t$ in $\bb R^3$. Then $U^{3k}=I$ and $U^k$ is the matrix of rotation by $2\pi/3$ about the same axis, namely \[U^k=\begin{pmatrix}
       0 & 0 & 1 \\
       1 & 0 & 0 \\
       0 & 1 & 0 \\
     \end{pmatrix}.\]
  Let $e_1=(1,0,0)^t$. For $x\in \bb Z_k$, $a\in \bb Z_3$ define vectors $v_{x,a}\in \bb R^3$ by
  \[v_{x,a}=U^{k(a-x)+x}e_1.\]
  Then we have \[ v_{x+1,a+1}=Uv_{x,a},\;v_{x,a+1}=U^kv_{x,a}.\]
  It follows that
  $\{v_{x,a}:a\in \bb Z_3\}$ is an orthonormal basis of $\bb R^3$ for each $x\in \bb Z_k$ and
  \[ \langle v_{x+1,a+1}|v_{x,a}\rangle=\langle Ue_1|e_1\rangle=\frac{1+2c_k}3.\]
  The lower bound for $\omega_{q_1}^s(C_k^d)$ follows by Corollary~\ref{cor:q1}.

  For lower bound for $\omega_{q_1}(C_k^d)$, let \[a=\frac{c_k}{\sqrt{1+2c_k^2}},\;b=\frac{(1-c_k)}{3\sqrt{1+2c_k^2}}\] and consider the unit vector
  \[ \eta=a\sum_{i=1}^3 e_i\otimes e_i+b\sum_{i,j=1}^3 e_j\otimes e_j\in \bb C^3\otimes \bb C^3.\]
  Let $p(a,b|x,y)=|\langle \eta|v_{x,a}\otimes v_{y,b}\rangle|^2$ be the corresponding $q$-density, where $v_{x,a}$ is as above. Using the fact that $U$ fixes $(1,1,1)^t$ and that $\langle Uv_{x,a}|v_{x,a}\rangle=\langle Ue_1|e_1\rangle=\frac13(1+2c_k)$, a calculation shows that the expected winning probability for $p$ is
  \begin{align*}
    \frac1k&\sum_{x\in \bb Z_k,a\in \bb Z_3}|\langle \eta|v_{x,a}\otimes v_{x+1,a+1}\rangle|^2=\frac1k\sum_{x\in \bb Z_k,a\in \bb Z_3}|\langle \eta|(1\otimes U)v_{x,a}\otimes v_{x,a}\rangle|^2\\&=\frac1k\sum_{x\in \bb Z_k,a\in \bb Z_3}|a\langle v_{x,a}|Uv_{x,a}\rangle+b|^2=\frac{(a(1+2c_k)+3b)^2}3=\frac{1+2c_k^2}3.
  \end{align*}

  For the upper bound on $\omega_{qc}^s(C_k^d)$, recall the notation of
  Theorem~\ref{thm:SOS+C}. For this game $(\cl G,\pi)$ we have $R_{\cl G,\pi}\in \bb C(\bb F(k,3))$,
  \[ R_{\cl G,\pi}=\frac1{2k}\sum_{x\in \bb Z_k,a\in \bb Z_3} e_{x,a}e_{x+1,a+1}e_{x,a}+e_{x+1,a+1}e_{x,a}e_{x+1,a+1}\]
  so that $R_{\cl G,\pi}\in R_0+\cl C$, where
  \[ R_0=\frac1{2k}\sum_{x\in \bb Z_k,a\in \bb Z_3}e_{x,a}e_{x+1,a+1}+e_{x+1,a+1}e_{x,a}.\]
  Let
  \[E=(e_{0,0}, e_{0,1},e_{0,2},...,e_{k-1,0},e_{k-1,1},e_{k-1,2})^t\in \bb C^{3k}\otimes \bb C(\bb F(k,3)).\]
  Let $S_n$ be the cyclic forward shift $e_j\mapsto e_{j+1}$ on $\bb C^n$; for example, we saw above that $U^k=S_3$. Let $W=S_k\otimes S_3\in M_{3k}(\bb C)$. We have
  \[ R_0= \frac1{2k} E^* (W+W^*) E.\]
  Moreover, $e_{x,a}e_{x,a'}=0$ whenever $a\ne a'$, so if $A$ is any $3\times 3$ matrix with zero diagonal, then $E^*(I_k\otimes A)E=0$. Let
  $A=\frac12(S_3+S_3^*)$
  and $\alpha,t\in \bb R$. Then
  \[ t1-R_0=\frac1k E^*XE\]
  where $X\in M_{3k}(\bb C)$,
  \begin{align*}
    X&=t I_{3k} + \alpha(I_k\otimes A)-\frac12(W+W^*)
    = \re( tI_{3k}+ \alpha I_k\otimes S_3-W).
  \end{align*}
  We have $I_k\otimes S_3=W^k$, and it follows that the eigenvalues of
  $X$ are the complex numbers
  \[ \lambda_j=t+\alpha\cos(2\pi j/3)-\cos(j\theta_k),\quad j\in \bb Z_{3k}.\]
  Hence, the smallest eigenvalue of $X$ is either
  \[ \lambda_0=t+\alpha-1\text{ or }\lambda_1=t-\alpha/2-c_k.\]
  So if $t= \max\{1-\alpha,\alpha/2+c_k\}$, then $X$ is positive semidefinite. The smallest $t$ occurs for $\alpha=2(1-c_k)/3$, with $t=(1+2c_k)/3$. Factoring $X=Y^*Y$, we obtain $t1-R_0=\frac1k(YE)^*(YE)\in SOS$. The upper bound for $\omega_{qc}^s(C_k^d)$ follows from Theorem~\ref{thm:SOS+C}.
\end{proof}

For the directed $4$-cycle, we can compute the $vect$-values of the directed graph game precisely.
  \begin{prop}
    We have \[\omega_{vect}(C_4^d)=\omega_{vect}^s(C_4^d)=2(1-1/\sqrt{3}).\]
  \end{prop}
  \begin{proof}
  The directed $4$-cycle game $(\cl G,\pi)$ has symmetries induced by $(x,y,a,b)\mapsto (x+1,y+1,a+1,b+1)$ and $(x,y,a,b)\mapsto (y,x,b,a)$. The first symmetry implies that $(x,y,a,b)\mapsto(x,y,a+1,b+1)$ and $(x,y,a,b)\mapsto (x+1,y+1,a,b)$ are also symmetries. Hence, there is an optimal density in $\cl C_{vect}$ which is invariant under these symmetries.

  For each $p\in\mathcal{C}_{vect}$, there exist vectors $v_{x,a}$,
  $w_{y,b}$ in some Hilbert space with
  $p(a,b|x,y)=\langle v_{x,a},w_{y,b}\rangle$,
  $\sum_a v_{x,a}=\eta=\sum_b w_{y,b}$ for some unit vector $\eta$,
  for all $x,y$, and the triples of vectors in each of these sums are
  mutually orthogonal. By symmetry, we may assume that, viewed as a $(\bb Z_4\times \bb Z_3)\times (\bb Z_4\times \bb Z_3)$ matrix, we have
  \[ p=\begin{pmatrix}
      X&Y&Z&Y^t\\
      Y^t&X&Y&Z\\
      Z&Y^t&X&Y\\
      Y&Z&Y^t&X
    \end{pmatrix}
  \]
  for some $3\times 3$ circulant matrices $X,Y,Z$ with entries in $[0,1]$ such that $X=X^t$ and $Z=Z^t$, and the sum of the entries of each of $X$, $Y$ and $Z$ is $1$. So $X,Y,Z\in \tfrac13\conv\{I_3,S_3,S_3^2\}$. From this block structure, we see that for $x,y\in \bb Z_4$ we have
  \begin{align*}
    X&=(\langle v_{x,a}|w_{x,b}\rangle)_{a,b},
    \\
    Y&=(\langle v_{x,a}|w_{x+1,b}\rangle)_{a,b},
    \\
    Y^t&=(\langle v_{x,a}|w_{x+3,b}\rangle)_{a,b},
    \\
    Z&=(\langle v_{x,a}|w_{x+2,b}\rangle)_{a,b}.
  \end{align*}
  Now let $v_{x,a}'=w_{x,a}'$ be given by
  \[ v_{x,a}'=w_{x,a}'=
    \begin{cases}
      v_{x,a}&\text{$x$ even,}\\
      w_{x,a}&\text{$x$ odd.}
    \end{cases}
  \]
  Let
  $p'=(\langle v'_{x,a}|v'_{y,b}\rangle)_{(x,a),(y,b)}\in \cl
  C_{vect}$ be the corresponding vector correlation. It follows that
  $p'$ is positive semidefinite and
  \[ p'
    =
    \begin{pmatrix}
      X_0&Y&Z_0&Y^t\\
      Y^t&X_1&Y&Z_1\\
      Z_0^t&Y^t&X_2&Y\\
      Y&Z_1^t&Y^t&X_3\\
    \end{pmatrix}
  \]
  for some diagonal matrices $X_x$, $x\in \bb Z_4$, and some
  $Z_0,Z_1\in M_3(\bb C)$. The point is that the matrix blocks $Y$ and $Y^t$ are
  unchanged from $p$, and the strategy value only depends on these blocks. We
  can average $p'$ over the game symmetries to get $q\in \cl C_{vect}$ with the
  same $Y$ blocks, and each diagonal block will be a diagonal
  circulant matrix with entries adding to $1$, so $\tfrac13 I_3$. Hence,
  \[ q=
    \begin{pmatrix}
      \tfrac13I_3&Y&W&Y^t\\
      Y^t&\tfrac13I_3&Y&W\\
      W&Y^t&\tfrac13I_3&Y\\
      Y&W&Y^t&\tfrac13I_3
    \end{pmatrix}\in  M_{12}(\bb C)^+\]
  for some symmetric circulant matrix $W$, and $\bb E_{\cl G,\pi}(q)=\bb E_{\cl G,\pi}(p)$. So
  $\omega_{vect}(\cl G)=\sup_q\bb E_{\cl G,\pi}(q)$ where the sup is over $q$ of this form (where $Y$ is a circulant matrix). %

  Let $u^*=2(1-\frac1{\sqrt3})\approx 0.845$. We first wish to show that $\omega_{vect}(C_4^d)\le u^*$.  Parameterize the matrices $Y$ and $W$ above as $Y=\tfrac13((1-u-v)I_3+uS_3+vS_3^2)$ for some non-negative $u,v$ with $u+v\le 1$, and $W=\tfrac13(rI_3+\frac{(1-r)}2(S_3+S_3^2))$, where $r\in [0,1]$ (recall that $W=W^t$). Note that $\bb E_{\cl G,\pi}(q)=u$. We will show that $q$ being positive semidefinite implies that $u\le u^*$.

  The $3\times 3$ blocks of $q$ are mutually commuting. The eigenvalues of $3Y$ are $\{1,\lambda,\overline{\lambda}\}$ where $\lambda=1-u-v+u\xi+v\xi^2=1-\frac32(u+v)+i\frac{\sqrt{3}}2(u-v)$, and the eigenvalues of $3W$ are $\{1,s,s\}$ where $s=(3r-1)/2$. Diagonalizing, we find that $3q\cong Q \oplus R\oplus \overline{R}$ where $Q$ is the $4\times 4$ matrix of ones, and $R$ is the $4\times 4$ circulant matrix with first row $(1,\lambda,s,\overline{\lambda})$, i.e., $R=I_4+sS_4^2+2Re(\lambda S_4)$. The eigenvalues of $R$ are
  \[
    1+(-1)^ks+2Re(i^k\lambda),\; k=0,1,2,3.
  \]
  These must all be non-negative for $q$ to be positive semidefinite.
  So $2|Re(\lambda)|\le 1+s$ and $2|Im(\lambda)|\le 1-s$.
  Recalling our paramaterization in terms of $u,v,r$, we obtain the necessary conditions
  \[ u+v\le \frac r2+\frac56,\quad u-v\le \frac{\sqrt{3}}2(1-r),\quad 0\le u,v,r\le 1.\]
  Let $r^*=\frac73-\frac4{\sqrt3}$. For each $r\in [0,1]$, the maximum value of $u$ turns out to be
  \[
    u_{max}(r)=\begin{cases}
      (5+3 r)/6&0\le r\le r^*\\
      (5+3\sqrt{3}-3(\sqrt{3}-1)r)/12&r^*\le r\le 1
    \end{cases}
  \]
  This clearly attains its maximum at $r=r^*$, so $\omega_{vect}(\mathcal{G})\le u(r^*)=u^*$.

  Finally, to see that $\omega_{vect}^s(C_4^d)\ge u^*$,
  let  $R_\theta$ be the $2\times 2$ matrix of rotation by $\theta$, for $(x,a)\in \bb Z_4\times \bb Z_3$ define unit vectors $s_{x,a},t_{x,a}\in \bb R^2$ by
   \[s_{0,0}=t_{0,0}=e_1,\quad s_{x+1,a+1}=R_{\pi/6}s_{x,a},\quad t_{x+1,a+1}=R_{\pi/3}t_{x,a},\]
   and consider the vectors $v_{x,a}\in \bb R^5$,
  \[ v_{x,a}=\frac13\begin{pmatrix}1\\ \alpha s_{x,a}\\\beta t_{x,a}
     \end{pmatrix}
   \]
   where $\alpha=\sqrt{2(\sqrt{3}-1)}$ and $\beta=\sqrt{2(2-\sqrt3)}$. Observe that
   \[ s_{x,a+1}=R_{\pi/6}^4s_{x,a}=R_{2\pi/3}s_{x,a},\quad
     t_{x,a+1}=R_{\pi/3}^4t_{x,a}=R_{4\pi/3}t_{x,a}.\] So
   \[\langle
    v_{x,a}|v_{x,a'}\rangle=\tfrac19(1+\alpha^2\cos(2\pi/3)+\beta^2\cos(4\pi/3))=\tfrac19(1-\tfrac12(\alpha^2+\beta^2))=0\]
   for $a\ne a'$.  Note that
   $\sum_{a}s_{x,a}=(I+R_{2\pi/3}+R_{4\pi/3})s_{x,0}=0$ and similarly
   $\sum_at_{x,a}=0$, so $\sum_a v_{x,a}=(1,0,0,0,0)^t$ for every
   $x$. And one can check by explicit calculation that
   $p(a,b|x,y):=\langle v_{x,a}|v_{y,b}\rangle$ has no negative
   entries. So $\{v_{x,a}\}$ defines a synchronous vect-density $p(a,b|x,y)$. Its
   winning probability for the cycle game is
   \begin{align*}\frac14&\sum_{x\in \bb Z_4,a\in \bb Z_3}p(a,a+1|x,x+1)=3p(0,1|0,1)\\&=3\langle
     v_{0,0}|v_{1,1}\rangle=\tfrac13(1+\alpha^2\cos(\pi/6)+\beta^2\cos(\pi/3))=u^*.\qedhere
   \end{align*}
 \end{proof}

\section{Generalizations of CHSH}

Let $X = Y = A = B = \bb Z_n$ and define $f: \bb Z_n \times \bb Z_n \to \bb Z_n$ by $f(x,y) = xy$, then we obtain a non-synchronous, symmetric group based game on the group $(\bb Z_n, +)$ with rules,
\[ V(x,y,a,b) =1 \iff a -b = xy,\]
which in the case $n=2$ is the CHSH game.  Note that in the definition of this game we are using that $\bb Z_n$ is a ring.

A non-synchronous generalization of the CHSH games, denoted $CHSH_n$, with $a+b=xy$, were introduced by Buhrman and Massey~\cite{BM} and further studied by Bavarian and Shor~\cite{BS}, who also introduced the variant with rules
\[ V(x,y,a,b) =1 \iff ab= x+y.\]

Note that if $p(a, b| x, y) \in C_t$ and we set $q(a, b |x, y) = p(a, -b | x,y)$, then $q \in C_t$. For this reason, the non-synchronous values of our version and the game $CHSH_n$ will be equal. In particular, the synchronous values of our version of the game will give lower bounds on the corresponding value of the non-synchronous game $CHSH_n$. 
 
Even in the case $n=3$ and the uniform distribution on input pairs, the quantum values of this game are not known exactly.  However using the methods developed in the section on directed graphs we obtain exactly the synchronous values of this non-synchronous game.
We begin with an extension to the non-synchronous case of  Theorem~\ref{thm:value1}.

Let $\xi = e^{2 \pi i/3}$, let $\sigma_3 = \{ 1, \xi, \xi^2 \}$, be the cyclic group of order 3 written multiplicatively, let $X$ be a set with $|X| = n$ and let
$f:X \times X \to \sigma_3$, let $\bb F(n,3)$ be the free product of $n$ copies of $\sigma_3$ and let $u_x, \, x \in X, \,\, u_x^3=1$ be the generators. Let $\cl G$ be the corresponding game and let $\pi$ be a density on $X \times X$.

\begin{thm} Let $f:X \times X \to \sigma_3$ be a function, let $\cl G=(X, \sigma_3, f)$ be the corresponding group based game and let $\pi$ be a density on inputs.  If we set
\[H = \frac{1}{3} + \frac{1}{3}  \sum_{x,y} \pi(x,y) [\xi^{f(x,y)} u_x^*u_y + \xi^{-f(x,y)}u_y^*u_x],\]
 then
\[\omega^s_{qc}(\cl G, \pi) = \sup \{ \tau(H)\mid  \tau: C^*(\bb F(n,3)) \to \bb C \text{ is a tracial state } \}.\] Consequently,
\[\omega^s_{qc}(\cl G, \pi) = \inf \{ t:  t1 -H \in SOS+ \cl C \},\]
in $\bb C(\bb F(n,3))$. 
\end{thm}
\begin{proof} We have that for any trace $\tau$,
  \begin{align*}
    \tau(\xi^lu_x^*u_y+\xi^{-l}u_y^*u_x)&=2\re\tau(\xi^l u_x^*u_y)\\&=\sum_{a,b\in \bb Z_3}2\re\tau(\xi^{l-a+b}e_{x,a}e_{y,b})
    \\&=\sum_{c\in \bb Z_3}2\re(\xi^c)\sum_{b\in \bb Z_3}\tau(e_{x,b+l-c}e_{y,b})
    \\&=\sum_{b\in \bb Z_3} 2\tau(e_{x,b+l}e_{y,b})-\tau((e_{x,b+l-1}+e_{x,b+l-2})e_{y,b})
    \\&=\sum_{b\in \bb Z_3} 2\tau(e_{x,b+l}e_{y,b})-\tau((1-e_{x,b+l})e_{y,b})
    \\&=-1+3\sum_{b\in \bb Z_3}\tau(e_{x,b+l}e_{y,b}).
  \end{align*}
  Thus, if $p(a,b|x,y) = \tau(e_{x,a} e_{y,b})$ we have that
  \begin{align*}
    \bb E_{\cl G,\pi}(p)&=\sum_{x,y\in X}\pi(x,y)\sum_{b\in \bb Z_3}\tau(e_{x,b+f(x,y)}e_{y,b})
    \\&=\frac13+\frac13\sum_{x,y\in X}\pi(x,y)\tau(\xi^{f(x,y)}u_x^*u_y+\xi^{-f(x,y)}u_y^*u_x)=\tau(H).
  \end{align*}
So\[ \omega_{qc}^s(\cl G, \pi) = \sup \{ \tau(H) : \tau \text{ tracial state } \}.\qedhere\]
\end{proof}

 We now return to the generalized CHSH game on the ring $\bb Z_3$.

 \begin{thm}
 Let $f: \bb Z_3 \times \bb Z_3 \to \bb Z_3$ be given by $f(x,y) = xy$, and consider the group based game $\cl G$ with rules,
\[ V(x,y,a,b) =1 \iff a -b = xy,\]
and $\pi$ the uniform distribution.
Then
\[\omega_{det}(\cl G, \pi) = 2/3<\omega_{q_1}(\cl G,\pi),\] and
\[\omega^s_{det}(\cl G, \pi) = \omega^s_{qc}(\cl G,\pi)=5/9.\]
\end{thm}
\begin{proof}
It is clear that $\omega_{det}^s(\cl G,\pi)$ is attained at any constant strategies, which yields the value $5/9$.
Similarly, $\omega_{det}(\cl G,\pi)=2/3$, is attained at (for example) $g,h:\bb Z_3\to \bb Z_3$, $g:0\mapsto 0,1\mapsto 1,2\mapsto1$ and $h=1-g$.

It follows from~\cite{JLLNL} (see also~\cite{KSTBSA}) that $\omega_{q_1}(\cl G,\pi)=(\frac{1+2\cos(2\pi/9)}3)^2 %
\approx 0.7124$. %

The upper bound for $\omega_{qc}^s(\cl G,\pi)$ follows from Corollary~\ref{cor:cost-matrix}. Let
\[ C=\frac1{18}\left(
\begin{smallmatrix}
 2 & 0 & 0 && 2 & 0 & 0 && 2 & 0 & 0 \\
 0 & 2 & 0 && 0 & 2 & 0 && 0 & 2 & 0 \\
 0 & 0 & 2 && 0 & 0 & 2 && 0 & 0 & 2 \\[3pt]
 2 & 0 & 0 && 0 & 1 & 1 && 0 & 1 & 1 \\
 0 & 2 & 0 && 1 & 0 & 1 && 1 & 0 & 1 \\
 0 & 0 & 2 && 1 & 1 & 0 && 1 & 1 & 0 \\[3pt]
 2 & 0 & 0 && 0 & 1 & 1 && 0 & 1 & 1 \\
 0 & 2 & 0 && 1 & 0 & 1 && 1 & 0 & 1 \\
 0 & 0 & 2 && 1 & 1 & 0 && 1 & 1 & 0
\end{smallmatrix}\right),\quad Z=-\frac1{27}\left(
\begin{smallmatrix}
 4 & 0 & 0 && 0 & 0 & 0 && 0 & 0 & 0 \\
 0 & 4 & 0 && 0 & 0 & 0 && 0 & 0 & 0 \\
 0 & 0 & 4 && 0 & 0 & 0 && 0 & 0 & 0 \\[3pt]
 0 & 0 & 0 && -2 & 3 & 3 && 0 & 0 & 0 \\
 0 & 0 & 0 && 3 & -2 & 3 && 0 & 0 & 0 \\
 0 & 0 & 0 && 3 & 3 & -2 && 0 & 0 & 0 \\[3pt]
 0 & 0 & 0 && 0 & 0 & 0 && -2 & 3 & 3 \\
 0 & 0 & 0 && 0 & 0 & 0 && 3 & -2 & 3 \\
 0 & 0 & 0 && 0 & 0 & 0 && 3 & 3 & -2
\end{smallmatrix}\right)
.\]
Ordering $X\times A=\bb Z_3\times \bb Z_3$ lexicographically, we have $C=C_{\cl G,\pi}$ and $Z\in \cl Z(\cl G)$, so $\omega_{qc}^s(\cl G,\pi)\le |X|\lambda_{max}(C+Z)=5/9$, as required.
\end{proof}

\begin{remark}
  If we instead take $Z=0$ in the previous proof then we only obtain a trivial
  upper bound, $\omega_{qc}^s(\cl G,\pi)\le
  |X|\lambda_{max}(C)=1$. This illustrates the importance of considering the set
  $\cl Z(\cl G)$ in Corollary~\ref{cor:cost-matrix}.
\end{remark}

 \section{The Commutator Game}

 Let $\Omega$ be a finite group.  It is natural to study games that measure how far $\Omega$ is from a commutative group. To this end, let $X = Y = A = B =\Omega$ and consider the group based game obtained by setting $f(x,y) = xyx^{-1}y^{-1} = (xy)(yx)^{-1}$. This defines a synchronous, symmetric $\Omega$-based game. We have a perfect strategy for this game if and only if there is a function $w: \Omega \to \Omega$ such that
 \[ w(x) w(y)^{-1} = (xy)(yx)^{-1}.\]
 Note that $w(x)$ is a perfect strategy if and only if $w(x) w(e)^{-1}$ is a perfect strategy, so we may assume that $w(e) =e$.
 Hence
 setting $y=e$ we obtain, that there is a perfect strategy if and only if there is a function $w$ satisfying $w(x) = xex^{-1}e^{-1} w(e)=e$. Hence a perfect deterministic strategy exists if and only if
 \[ e= w(x)w(y)^{-1} = xyx^{-1} y^{-1}, \forall x,y,\]
 which is if and only if $\Omega$ is abelian.

 So when $\Omega$ is non-abelian even the synchronous deterministic value is a potentially interesting measure of the level of non-commutativity.

Let us take the uniform distribution on $\Omega \times \Omega$, again we can assume that $w(e) =e$ and for the synchronous value we are trying to maximize the number of times that
  \[ w(x)w(y)^{-1} = xyx^{-1} y^{-1}.\]

  If one uses the deterministic strategy, $w(x) = e, \forall x$, then the value of this strategy multiplied by $|\Omega|^2$ is just a count of the number of pairs of elements that commute, a classically computed number for a group. In fact, the number of elements that commute is a multiple of the number of conjugacy classes in the group, often denoted $k(\Omega)$.

  It is natural to expect that $w(x) =e$ is the optimal deterministic strategy.  But calculations show that this is false!

Consider the case that $\Omega=S_3$ the group of permutations of 3 elements.  This group has 6 elements, and of the $36$ pairs of group elements, $18$ are the identity. Thus, if we take the deterministic strategy $w(x) = e$, then the expected value is 1/2.

  For generators of this group, let us use the cyclic permutation and the flip:
  \[ g: 1 \to 2 \to 3 \to 1, \,\, f: 1 \to 1, 2\to 3 \to 2.\]

  Consider the deterministic strategy given by:
  \[w(e) = w(f) = w(fg) = w(fg^2) =e, \,\, w(g) = g^2, \,\, w(g^2) =g.\]
  Then the number of times that $w(x) w(y)^{-1} = xyx^{-1} y^{-1}$ is 24.
  An exhaustive computer search shows that this is an optimal deterministic
  strategy.

  For a general finite group, we have no idea how to find the function $w$ that maximizes the cardinality of the set
  \[ \{ (x,y) :  w(x) w(y)^{-1} = xyx^{-1} y^{-1}, \,\, x,y \in \Omega \}, \]
  or what this number means as a parameter of the group.

  A similar problem exists if instead of $f(x,y) = xy x^{-1} y^{-1}$ one considers the game defined by $f(x,y) = yx$.

 Once one is able to compute this value for groups, it would be interesting to see if a quantum advantage exists for this game and attempt to find an interpretation of this ``quantum'' count of commutators.
 
 So we end with some questions.
 
 \begin{prob} Given a finite group, $\Omega$, can one identify $|\Omega|^2 \omega_{det}^s(\cl G, \pi)$ in group-theoretic terms? Are there groups that exhibit a quantum commutativity advantage, i.e., for which $\omega^s_{qc}(\cl G, \pi) > \omega^s_{det}(\cl G, \pi)$?
 \end{prob}

\end{document}